\documentclass[12pt]{article}
\usepackage{amsthm,amsmath,amsfonts,amssymb}
\usepackage{amsmath,amsfonts,amsthm}
\usepackage{color}
\usepackage{comment}
\usepackage{enumerate}
\newcommand{\R}{{\mathord{\mathbb R}}}
\newcommand{\Z}{{\mathord{\mathbb Z}}}
\newcommand{\N}{{\mathord{\mathbb N}}}
\newcommand{\C}{{\mathord{\mathbb C}}}

\def\lam {\lambda}

\newcommand{\HH}{\mathcal{H}}

\newcommand{\FF}{\mathcal{F}}

\newcommand{\hh}{\mathfrak{h}}

\newcommand{\ran}{{\rm Ran}}

\newcommand{\inn}[1]{\langle {#1} \rangle }
\newcommand{\abs}[1]{\left| #1 \right|}

\newcommand{\muu}[1]{\begin{align*} #1 \end{align*}}

\def\one{{\sf 1}\mkern-5.0mu{\rm I}}

\newcommand{\ben}{\begin{displaymath}}
\newcommand{\een}{\end{displaymath}}
\newcommand{\beqn}{\begin{equation}}
\newcommand{\eeqn}{\end{equation}}
\newcommand{\beqna}{\begin{eqnarray*}}
\newcommand{\eeqna}{\end{eqnarray*}}

\def\inf{{\rm inf}\,}

\newcommand{\sfrac}[2]{\textrm{\footnotesize $\frac{#1}{#2}$}}

\newtheorem{lemma}{Lemma}
\newtheorem{theorem}[lemma]{Theorem}
\newtheorem{remark}{Remark}
\newtheorem{proposition}[lemma]{Proposition}

\newtheorem{hyp}{Hypothesis}

\title{Ground States for translationally invariant Pauli-Fierz Models
at zero Momentum}

\author{David Hasler \and Oliver Siebert}

\date{\small Department of Mathematics, Friedrich-Schiller University of Jena \\
Jena, Germany}

\begin{document}

\maketitle

\begin{abstract}
We  consider the  translationally  invariant Pauli-Fierz model describing
a charged particle  interacting with the electromagnetic field. We show under
natural assumptions   that   the fiber Hamiltonian at zero
momentum has  a ground state.
\end{abstract}

\section{Introduction}

Non-relativistic qed has been a successful theory describing low energy aspects
of quantum mechanical  matter interacting with  the quantized radiation field.
Within this model many physical phenomena have been mathematically understood.
In the present paper, we shall discuss an aspect which falls within
 the scattering problem  of an electron interacting
with the quantized radiation field.
The main result which we prove is that the the total system composed of electron and photon field
at  zero momentum is  free of infrared divergences. In this situation the infrared singularity is critical. As soon as the momentum is nonzero a ground stated ceases to exist, \cite{haslerherbst1}.

Starting with the work of  Bloch and Nordsieck \cite{blonor:37} the so called
infrared catastrophe in scattering theory has been intensively investigated. Although
the physical reasons for infrared divergences were  well understood at that time  and   did
not lead to  any physical problems, the formal treatment was not  satisfactory.
In \cite{kulfad:70}  Fadeev and Kulish   introduced a new space of
asymptotic states and gave a theoretical discussion of this phenomenon
in the framework of relativistic quantum field theory.
   Fr\"ohlich
studied  such asymptotic
 states in  the so called Nelson Model    \cite{F73,F74}, which is mathematically well
defined. Extending these   results an iterative algorithm for constructing
 asymptotic  states in Nelson's model was developed in \cite{piz05}.
 This construction was extended to the model of non-relativistic qed
in  \cite{ChenFroehlichPizzo2}. Specifically,  the case of zero momentum  in nonrelativistic
qed has  previously been investigated  in   \cite{BCFS05,C01}.

In the present  paper we consider the Hamiltonian $H$ of non-relativistic qed
describing an electron, with or without spin, coupled to the quantized
radiation field.  The Hamiltonian commutes with the operator of
total momentum. We are interested in the operator, $H(0)$, obtained by restricting
the Hamiltonian   to the subspace of  total momentum zero.
Based on  a natural  energy inequality, we prove that for all values
of the coupling constant, $H(0)$ has a ground state. The energy inequality has been
shown to hold in the spinless case for all values of the coupling constant \cite{G72}, and in the
case of spin the energy inequality follows in a related situation from the main theorem in  \cite{C01}.
The existence of  a ground state can  for example be   used to obtain expansions on the binding energy  of the hydrogen atom  \cite{barchevouvou:10}.
To the best of our knowledge  the  result, which we prove, has so far only  been shown in the  spinless case   
 for  small values of the coupling constant  \cite{BCFS05},  see  also \cite{ChenFroehlichPizzo2} for related results.
In contrast to the proofs given there,   our proof  is  non-perturbative and   independent
of the magnitude  of any ultraviolet cutoff  parameter.  The  proof, which we give,  uses a compactness argument and  is  not constructive. Nevertheless,
 once the   existence of the ground state  is established one can use other methods to obtain
   asymptotic expansions  of the ground state as well as its energy   \cite{arai:14,brahaslan:18}, in this context  see also \cite{haisei02}.

The idea of the  proof given in the present paper follows closely the ideas   introduced  in   \cite{GriLieLos:01}, which were applied in a similar case in \cite{LMS06}. However,  in the  situation
which we encounter the infrared singularity is more severe
and subtler estimates are necessary.
 We believe that the estimates  which we use in the present paper might establish  alternative proofs  of  existence of ground states in  other critical cases, as for example  \cite{hasher:11}.

In the next section  we introduce the model and state the main result.
The proofs are presented in Section 3.

\section{Model and Statement of Results}

\label{sec:nonqed} \label{sec:eldef}

Let   $\hh$ be a complex Hilbert space.  We introduce the symmetric Fock space
$$
\FF(\hh)  =  \bigoplus_{n=0}^\infty  \hh^{(n)} ,
$$
where $\hh^{(0)} = \C$ and where  $\hh^{(n)} =  \mathcal{P}_n( \bigotimes_{k=1}^n \hh ) $    for $n \geq 1$,
with   $\mathcal{P}_n$ denoting  the orthogonal projection onto the subspace of
totally symmetric tensors. Thus we can identify
$\psi  \in \FF(\hh) $ with the sequence $(\psi_{(n)})_{n \in \N_0}$ with $\psi_{(n)} \in \hh^{(n)}$.
The vacuum  is the vector $\Omega :=
(1,0,0, \ldots ) \in \FF(\hh)$.
We define for $f \in \hh$ the creation operator $a^*(f)$ acting on vectors $\psi \in \FF$ by
$$
(a^*(f) \psi)_{(n)} = \sqrt{n} \mathcal{P}_n ( f \otimes \psi_{(n-1)} )
$$
with domain $D(a^*(f)) := \{ \psi \in \FF : a^*(f) \psi \in \FF \}$. This  yields a densely defined
closed operator. For $f \in \hh$ we define the annihilation $a(f)$ as the adjoint of $a^*(f)$, i.e.,
$$
a(f)  = \left[a^*(f) \right]^* .
$$
 It follows from the definition  that
 $a(f)$ is anti-linear, and $a^*(f)$ is linear in $f$. Creation and annihilation operators  are well
 known to satisfy the so called canonical commutation relations
$$
[ a^*(f) , a^*(g) ] = 0 \quad , \quad [a(f) , a(g) ] = 0 \quad ,
\quad [a(f) , a^*(g) ] =  \inn{ f , g }_\hh \; ,
$$
where $f , g \in \hh$, $[\cdot, \cdot]$ stands for the commutator, and    $\inn{ f , g}_\hh$ denotes the
inner product of $\hh$. For a self-adjoint operator  $A$  in $\hh$ we define  the operator $d \Gamma(A)$ as
follows.   In $\hh^{(n)}$  we define
 $$
A^{(n)} := A \otimes \one \otimes \cdots \otimes  \one  + \one \otimes A \otimes \cdots \otimes \one + \cdots + \one \otimes \cdots \otimes \one \otimes A ,  \, n \in \N ,
$$
in the sense of \cite[VIII.10]{rs1} and  $A^{(0)} := 0$.
By definition   $\psi \in \FF(\hh) $ is  in the domain of
$d \Gamma(A)$  if  $\psi_{(n)} \in D(A^{(n)})$ for all  $n \in \N_0$  and
\begin{align} \label{eq:defofgammaA}
(d \Gamma(A) \psi )_{(n)} = A^{(n)} \psi_{(n)}   ,  \quad n \in \N_0 ,
\end{align}
is  a vector in $\FF(\hh)$,  in which  case $d \Gamma(A) \psi $ is defined by \eqref{eq:defofgammaA}.
The operator $d \Gamma(A)$ is self-adjoint, see for example \cite[VIII.10]{rs1}.

Henceforth, we shall consider specifically
\begin{equation} \label{eq:defofhilconc}
\hh := L^2(\Z_2  \times \R^3 ) \cong L^2(\R^3 ; \C^2 )
\end{equation}
and write $\FF$ for $\FF(\hh)$.
The Hilbert space $\hh$ describes a so called  transversally polarized photon. By physical interpretation
the variable $(\lambda,k) \in  \Z_2  \times \R^3$ consists of the wave
vector $k$ and the polarization label  $\lambda$.
Because of   \eqref{eq:defofhilconc}, the   elements  $\psi \in \FF$ can be identified
with sequences $(\psi_{(n)})_{n=0}^\infty$ of so called $n$-photon wave
functions, $\psi_{(n)}
\in L^2_{\rm sym}(( \Z_2 \times \R^3)^n)$, where the subscript ``${\rm sym}$'' stands for the subspace
of functions wich  are totally symmetric in their $n$ arguments.
Henceforth, we shall make  use of  this identification without mention.
The Fock space inherits a scalar
product from $\hh$, explicitly
$$
\inn{  \psi , \varphi }  =  \overline{\psi}_{(0)} \varphi_{(0)} +
\sum_{n=1}^\infty \sum_{\substack{ \lam_1, \ldots ,\lam_n \\ \in \{ 1 , 2 \} } } \int
\overline{\psi_{(n)}(\lam_1, {k}_1, \ldots , \lam_n, {k}_n)}
\varphi_{(n)}(\lam_1, {k}_1, \ldots , \lam_n, {k}_n ) d {k}_1 \ldots d
{k}_n \; .
$$

We shall make use of the physics notation of the creation and annihilation operators.
One defines for $(\lambda,k)\in \Z_2\times \R^3$  and $\psi \in \FF$
\begin{align} \label{eq:defofannihdir}
[a_\lambda(k) \psi]_{(n)}(\lam_1, {k}_1, \ldots , \lam_n, {k}_n)  & = \sqrt{n+1} \psi_{(n+1)}(\lambda, k , \lam_1, {k}_1, \ldots , \lam_n, {k}_n)  , \, n \in \N_0 .
\end{align}
Now  \eqref{eq:defofannihdir} defines a well defined operator $a_\lambda(k)$ on
$$D_\mathcal{S} := \{ \psi \in \FF : \psi_{(n)}  = 0 \text{ for all but finitely many } , \psi_{(n)} \in   \mathcal{S}((\Z_2 \times \R^3)^n) \} ,$$
where $\mathcal{S}$ stands for the Schwartz space.  In the sense of quadratic forms on  $D_\mathcal{S} \times D_\mathcal{S}$ its adjoint $a_\lambda^*(k)$ is well defined. Furthermore,  in the sense of quadratic forms one has for all $f \in \hh$ the identities
\begin{align*}
a(f)  & = \sum_{\lambda=1,2} \int     \overline{f(\lam, {k}) } a_\lam({k}) d{k} \quad , \\
a^*(f) &  = \sum_{\lambda=1,2} \int f(\lam, {k})
a^*_\lambda(k) d{k} \; .
\end{align*}
For details we refer the reader to  \cite[X.7]{ressim:fou}.
The field energy operator denoted by $H_f$
is given by
$$
H_f = d \Gamma(\omega)  ,
$$
where $\omega : \Z_2 \times \R^3 \to \R$ is defined by  $\omega(\lambda,k) = |k|$.
 The operator of momentum $P_f$ is defined as a three dimensional vector of operators, where
 the $j$-th component is defined by
 $$
 (P_f)_j := d \Gamma(\pi_j) ,
 $$
 where $\pi_j : \Z_2 \times \R^3 \to \R$ is defined by  $\pi_j(\lambda,k) = k_j$.

 The Hilbert space describing the
system composed of a charged particle with  spin $s \in \{ 0 , \frac{1}{2} \}$  and the quantized field is
$$
\HH_0 := L^2( \R^3 \times \Z_{2 s + 1} )  \otimes \FF \; .
$$
The Hamiltonian is
$$
H =   \frac{1}{2}\left(  p  + e A (x)  \right)^2  + e  S \cdot  B(x)  + H_f \; ,
$$
with
\begin{align} \label{eq:44}
A(x) & = \sum_{\lambda=1,2} \int \frac{ \varepsilon_{\lambda}(k) }{\sqrt{2|k|}}\left(
\overline{\rho(k)} a_{\lambda}(k) e^{ik \cdot x}  +
\rho(k) a_{\lambda}^*(k) e^{-ik \cdot x}  \right) dk
\; ,  \\
B(x) & = \sum_{\lambda=1,2} \int \frac{  i k \wedge \varepsilon_{\lambda}(k)  }{\sqrt{2|k|}} \left(
\overline{\rho(k)} a_{\lambda}(k) e^{ik \cdot x}  -
\rho(k) a_{\lambda}^*(k) e^{-ik \cdot x} \right) dk
\; ,
\end{align}
where the $\varepsilon_{\lambda}(k) \in \R^3$ are vectors, depending
measurably on $\widehat{k} = k /|k|$,  such that $(k/|k|,
\varepsilon_{1}(k), \varepsilon_{2}(k))$ forms an orthonormal basis. For  the proof
we shall  make  an explicit choice in   \eqref{eq:cheps}, below. We shall adopt the standard convention that for  $v = (v_1,v_2,v_3)$  we write $v^2 :=   \sum_{j=1}^3 v_j v_j $.
By $x$ we denote the position of the electron and its canonically
conjugate momentum by $p = - i \nabla_x$.
If $s=1/2$, let $S = (\sigma_1,\sigma_2,\sigma_3)$   denote the vector of  Pauli-matrices.
If $s=0$, let $S = 0$.
The number $e \in \R$ is called the coupling constant.
The so called form factor $\rho : \R^3 \to \C$ is a measurable function for which we shall  assume  the following
 hypothesis for the main theorem. 
\begin{hyp}\label{h:formfac1}   For some  $0 < \Lambda < \infty$ we have
\begin{equation}  \label{e:formfac1}
\rho(k) = \frac{1}{(2 \pi)^{3/2}} \chi_\Lambda(|k|) \; ,   \quad k \in \R^3 ,
\end{equation}
where $\chi_{\Lambda} = 1_{[0,\Lambda]}$.
\end{hyp}
We note that Hypothesis A is usually assumed in Pauli-Fierz type models.
The
Hamiltonian is translation invariant and commutes with the generator
of translations, i.e., the operator of total momentum
$$
P_{\rm tot} = p + P_f  \; .
$$
Let
$$
W = \exp(i x \cdot P_f) \; .
$$
Note $W P_{\rm tot} W^* = p $ so that in the new representation  $p$
is the total momentum. One easily  computes
$$
W H W^* =  \frac{1}{2}\left( p - P_f  +  e A \right)^2   + e  S \cdot  B         +      H_f \; ,
$$
where set $A := A(0)$ and $B := B(0)$. Let $F$ be the Fourier transform in the electron
variable $x$, i.e., on $L^2(\R^3)$,
\begin{eqnarray}
\label{eq:fourier} ( F \psi)(\xi ) = \frac{1}{(2\pi)^{3/2}}
\int_{\R^3} e^{-i \xi \cdot x } \psi(x) dx \; .
\end{eqnarray}
Then the composition $U=FW$ is a  unitary operator 
\begin{align*}
 U : \HH_0 \to
L^2(\R^3  \times \Z_{2s+1}) \otimes \FF \cong \int_{\R^3}^\oplus \C^{2s+1} \otimes   \FF d\xi,
\end{align*} 
 yielding  the so called fiber decomposition of the
Hamiltonian,
\begin{align*}
U H U^* = \int_{\R^3}^{\oplus} H(\xi) d\xi , 
\end{align*} 
where 
$$ H(\xi) = \frac{1}{2}\left( \xi - P_f  +  e A
\right)^2  + e  S \cdot  B    + H_f
 \;
$$
is an operator in  the so called  reduced  Hilbert space  $$\HH :=  \C^{2s+1} \otimes \FF .$$ The operator   $H(\xi)$ is self-adjoint on
 $D( P_f^2) \cap D(H_f )$, see  Theorem   \ref{hyp:op0} in the next section.
To prove that $H(0)$ has a ground state, we use a  compactness argument
similar to  \cite{GriLieLos:01}.   The idea
is to first introduce a  positive photon mass  in the field energy. For $m \geq 0$ we define
$$
H_{f,m} = d \Gamma(\omega_m) ,
$$
where $$\omega_m(\lambda, k) = \sqrt{m^2 + k^2} , $$
and  study the operator
\begin{equation} \label{def:ofhm}
H_m(\xi)  = \frac{1}{2}(\xi - P_f + eA)^2 + e S \cdot B + H_{f,m}.
\end{equation}
We  set
$$
E_m(\xi)  := \inf \sigma(H_m(\xi)) .
$$

The proof is based on the following energy inequality.
Specifically we can show our  result for any $e \in \R$ for which there exists an $m_0 > 0$ such that
   for all $m \in (0,m_0)$
 \begin{eqnarray}
  E_m(\xi) \geq E_m(0) \quad , \quad \forall \xi \in \R^3 .
\label{eq:eineq}
\end{eqnarray}
Inequality \eqref{eq:eineq} has been investigated  in the literature.
In spinless case, $s = 0$,   Inequality \eqref{eq:eineq} has in fact been shown  for all values
 $m \geq 0$ and   $e \in \R$
using   functional integration,  \cite{G72,SPOHN04,H06, LMS06}.
In case  of spin $s=1/2$  Inequality   \eqref{eq:eineq}    has  to the best
of our knowledge  not yet been  shown by means of  functional integration.
 For $s=1/2$  an   Inequality of the type  \eqref{eq:eineq} follows  in a related 
situation for  small $|e|$  from the main theorem
stated in \cite{C01}, which in turn  is based on  perturbative arguments. We now state 
the main result of this paper.

\begin{theorem} \label{thm:main1}  Suppose Hypothesis A  holds, and let   $e \in \R$.   If   there exists an $m_0 > 0$ such that
for all  $m \in (0,m_0)$ the  energy inequality  \eqref{eq:eineq} holds, then
     the operator
$H(0)$ has a ground state.
\end{theorem}

By the results in the literature mentioned in the previous paragraph,  we obtain
immediately the following  theorem as  corollary.

\begin{theorem} \label{cor:main1} Suppose   Hypothesis A holds. In the spinless case, i.e.,  $s=0$, the operator $H(0)$ has a ground state for all values of the coupling constant $e$.
\end{theorem}

We note that for small values of the coupling constant the statement  in   Therorem  \ref{cor:main1}
  has  been shown  previously in   \cite{BCFS05}, see also  \cite{ChenFroehlichPizzo2} for related work.
In the present paper  we extend that result    to  all  values of the coupling constant
and provide a rather short proof.
We want to point out,  that  the question  whether    Inequality \eqref{eq:eineq} holds
for all values of the coupling constant
in the case of spin $s=1/2$  seems to be an   open question.
We would like to  mention   work  \cite{HL08} in that direction.

\section{Proof of Results} \label{sec:proofelec}

We first state the following technical result about the domain of  self-adjointness. 

\begin{theorem} \label{hyp:op0} Let  $m \geq 0$ and $\rho \in L^2(\R^3;(|k| +\omega_m(k)^{-1} |k|^{-1}) dk)$. Then for all    $\xi \in \R^3$
and    $e \in \R$  the operator   $H_m(\xi)$ is
 self-adjoint on the natural domain of  $ \frac{1}{2} P_f^2 + H_{f,m}$.
\end{theorem}

Versions of this theorem have been shown in   \cite{H06},  \cite{HL08} and  \cite{LMS06}.
Since we could not find  the precise version, which we need, in the literature,   we shall provide  a short proof of Theorem \ref{hyp:op0} in Appendix \ref{appessself}. The proof  follows closely a proof given in \cite{HH08}.
Moreover, we shallv use  the following result to prove the main theorem. 
\begin{theorem} \label{hyp:op1}
Let  $\rho \in  L^2(\R^3;(|k| + |k|^{-1}) dk)  $ with $ \rho = \rho(- \cdot)$. Let $e \in \R$ and $m > 0$ and
suppose  \eqref{eq:eineq} holds.
If   $|\xi| \leq 1$, then   $E_m(\xi)$ is an eigenvalue of   $H_m(\xi)$  isolated from the
essential  spectrum.
\end{theorem}
Versions of this theorem have been shown in the literature  \cite{F74,F73,SPOHN04,GriLieLos:01,LMS06}.
We could not find in the literature  the precise version, which we need,  so we  provide  a  proof of Theorem \ref{hyp:op1} in Appendix \ref{appexreg}.

Let us now outline the strategy of the proof of Theorem  \ref{thm:main1}, which  follows closely ideas  given  in \cite{GriLieLos:01}.
For $m > 0$   it follows from  Theorem \ref{hyp:op1} that $E_m(0)$ is an eigenvalue  of $H_m(0)$.
Henceforth  let  $\psi_m$ denote  a normalized  eigenvector of $H_m(0)$
with eigenvalue $E_m(0)$.  We will show in  Proposition \ref{eq:propenegyconv2}, below,  that  $(\psi_m)_{m > 0}$ is a minimizing family for $H(0)$ as $m$ tends to zero, i.e.,
\begin{equation} \label{eq:quadform0}
0 \leq   \langle{ \psi_m ,  ( H(0) - E(0) ) \psi_m }\rangle \to 0  \quad ( m \downarrow 0 ) .
\end{equation}
Finally we shall use a compactness argument, based on two infrared bounds, stated in  Lemmas \ref{lem:first} and \ref{lem:second}, to show that there exists a strongly convergent subsequence $(\psi_{m_j})_{j \in \N}$  which
converges to a nonzero vector, say $\psi_0$. Using lower semicontinuity of nonnegative quadratic forms \cite{Sim:77}  (or alternatively the spectral theorem and  Fatou's Lemma),
it will then follow from \eqref{eq:quadform0}
that
\begin{eqnarray*}
 0 \leq \inn{  \psi_0 , ( H(0)  - E(0) ) \psi_0 }  \leq \liminf_{i \to \infty} \inn{  \psi_{m_i} , ( H(0) - E(0) ) \psi_{m_i} } = 0 ,
\end{eqnarray*}
i.e., that $\psi_0$ is a ground state of $H(0)$.

\begin{remark} {\rm We note that in contrast to \cite{GriLieLos:01, LMS06} the infrared bounds which we obtain
have   stronger  infrared singularities. 
 Therefore it is harder to prove compactness. This difficulty will be  addressed in  Lemma~\ref{thm:compact} below.}
\end{remark}

\subsection{Ground State Properties for  massive Photons}

Throughout  this section we assume that  $\rho \in L^2(\R^3;(|k| + |k|^{-2}) dk)    $.
 We will use the notation $N = d \Gamma(1)$.

\begin{proposition} \label{eq:propenegyconv}
We have  $E_m(0)  \geq E(0)$ and
\begin{eqnarray} \label{eq:enegineq1}  E(0) = \lim_{m \downarrow  0} E_m(0) \; . \end{eqnarray}
\end{proposition}
\begin{proof}
For $0 \leq m' \leq m$ we have  $\omega \leq \omega_{m'} \leq \omega_{m}$
and hence   $H(0) \leq H_{m'}(0) \leq H_{m}(0)$.
It follows that   $E_m(0)$ is monotonically decreasing as $m$ tends to zero and $E(0) \leq E_m(0)$.
This implies the existence of the limit and 
\begin{equation} \label{eq:emconv}
 \lim_{m \downarrow 0} E_m(0)   \geq E(0)   .
\end{equation}
To show the opposite inequality we argue as follows. From Theorem  \ref{hyp:op0}
it follows that any core for $P_f^2 + H_f$ is a core for $H(0)$.
Thus for any $\epsilon
>0$, there exists a normalized vector $\phi \in D(N) \cap D(P_f^2 + H_f)$ such that
$$
\inn {\phi, H(0) \phi } \leq E(0) + \epsilon \; .
$$
On the other hand, since  $H_m(0) \leq H(0) + m N$, it follows that for any
$m \geq 0$,
$$
E_m(0) \leq \inn{\phi, H_m(0) \phi} \leq \inn{\phi, H(0) \phi } + m \inn{\phi , N
\phi} \leq E(0) + \epsilon + m \inn{\phi, N \phi} \; .
$$
Hence
\begin{equation} \label{eq:emconv2}
 \lim_{m \downarrow  0} E_m(0) \leq E(0) + \epsilon .
\end{equation}
 Since
$\epsilon$ is arbitrary,   \eqref{eq:enegineq1}  follows from \eqref{eq:emconv} and  \eqref{eq:emconv2}.
\end{proof}

Let us collect a basic inequality in the following lemma.

\begin{lemma} \label{cor:1}  Let $e \in \R$ and  $m \geq 0$, and  suppose   \eqref{eq:eineq} holds.
Then for all  $\xi \in \R^3$ we have
$$H_m(\xi)  - E_m(0) \geq  0 . $$
\end{lemma}
\begin{proof} Follows from $H_m(\xi) \geq E_m(\xi) \geq E_m(0)$.
\end{proof}

For later use we state the following Proposition.

\begin{proposition} \label{eq:propenegyconv2}
Assume   $\rho(-\cdot) = \rho$. Let $e \in \R$ and suppose there exists an $m_0 > 0$ such that
 \eqref{eq:eineq} holds for  all $m \in (0,m_0)$.
Then
\begin{equation} \label{eq:quadform}
0 \leq   \langle{ \psi_m ,  ( H(0) - E(0) ) \psi_m }\rangle \to 0  ,
\end{equation}
in the limit $m \downarrow 0$.
\end{proposition}
\begin{proof} Using that $H_m(0) \geq H(0)$ we find from Proposition \ref{eq:propenegyconv} that
\begin{align*}
0 & \leq   \langle{ \psi_m ,  ( H(0) - E(0) ) \psi_m \rangle} \\
&  \leq   \langle{ \psi_m ,  ( H_m(0) - E(0) ) \psi_m \rangle} \\
&  = E_m(0) - E(0) \to 0  ,   \quad ( m \downarrow 0 ) .
\end{align*}
\end{proof}

\subsection{Infrared Bounds}

Throghout this section we assume that   $\rho \in L^2(\R^3;(|k|+|k|^{-1}) dk)$ with $\rho = \rho(- \cdot)$.
To simplify the notation we  write
\begin{align*}
v   &:=  -  P_f + e A  \; , \\
 h_m  & := H_m(0) , \\
e_m  &:= E_m(0)  .
\end{align*}

\begin{lemma} \label{lem:1}  Let $e \in \R$ und $m > 0$, and suppose
  \eqref{eq:eineq} holds. Then
for  $i=1,2,3$, the vector
$v_{i} \psi_m \in \FF$ is orthogonal to
$\psi_m$ and
$$
0 \leq \inn{ v_{i} \psi_m , ( H_m(0) - E_m(0) )^{-1}   v_{i} \psi_m } \leq
\frac{1}{2}
$$
\end{lemma}
\begin{proof}
For the proof we use analytic perturbation theory. For details we refer the reader to  \cite{kato,ressim:ana}.
On $\C^3$ the operator valued function $\zeta \mapsto
H_m(\zeta) = H_m(0)  + \zeta \cdot v +
\frac{1}{2}\zeta^2$ is an analytic family of type (A) in each component.
By  Theorem  \ref{hyp:op1}   we
know that  $E_m(0)$ is an eigenvalue isolated from the essential spectrum.
Let $P_m(0)$ be the
orthogonal projection onto the finite dimensional eigenspace of $E_m(0)$. By first order
perturbation theory 
 and the energy inequality \eqref{eq:eineq} we
conclude that $P_m(0) v P_m(0) = 0$. By second order perturbation theory
and an application of the energy inequality \eqref{eq:eineq}
we conclude that  for 
 $i=1,2,3$,
$$
0 \leq  \partial_{\xi_i} \partial_{\xi_i}  E_m(\xi)
\Big\vert_{\xi=0} \leq \left( 1  - 2 P_m(0) v_{i} ( H_m(0) - E_m(0))^{-1}
v_{i} P_m(0) \right) \; ,
$$
where the second inequality is understood  as  an operator inequality on $\ran P_m(0)$.  The second inequality   in fact   holds, since  $E_m(\xi)$ is defined as an  infimum.
This shows the claim.
\end{proof}

For notational convenience we set
\begin{eqnarray*}
R_m(k) := ( H_m(-k) + \omega_m(k) - e_m   )^{-1} \; , \quad k \in \R^3 ,
\end{eqnarray*}
which by Lemma   \ref{cor:1} is well defined and satisfies
\begin{equation} \label{eq:boundonres}
\| R_m(k) \| \leq \omega_m(k)^{-1} .
\end{equation}
The formula of the next Lemma is known as
a so called  Pull-through relation.
\begin{lemma}  \label{lem:pull}  Let $e \in \R$ and $m > 0$, and suppose
 \eqref{eq:eineq} holds.
Then for  a.e. $k$,
$a_\lambda(k) \psi_m \in \FF$ and
\begin{eqnarray} \label{eq:pull}
a_{\lambda}(k) \psi_m =\frac{e \rho(k)}{\sqrt{2|k|}} R_m(k) ( -
\varepsilon_{\lambda}(k) \cdot v + S  \cdot ( i k \wedge \varepsilon_\lambda(k)) )\psi_m  .
\end{eqnarray}
\end{lemma}

\begin{proof}
The proof is similar to \cite[Lemma 6.1]{chenfroehlich}, see also  \cite[Lemma 7]{haslerherbst1}.
By Theorem  \ref{hyp:op0}, we know that  $\psi_m \in D(H_f)$. Hence using the standard expression
of the free field energy in terms of annihilation operators
$$
\sum_{n=0}^\infty \sum_{\lambda} \int |k|  \left\| ( a_{\lambda}(k)
\psi_m )_{(n)} \right\|^2 dk  = \inn{ \psi_m , H_f \psi_m } < \infty \; ,
$$
which implies $a_\lambda(k) \psi_m \in \FF$ for  a.e. $k$. We write
\begin{equation} \label{eq:defoffieldfunc}
f_A(k,\lambda) := \frac{\rho(k)}{\sqrt{2| k|}} \varepsilon_\lambda(k), \qquad f_B(k,\lambda) := - i k \wedge f_A(k,\lambda) . 
\end{equation}
By the canonical  commutation relations of creation and annihlation operators  we find
\muu
{
a_\lambda(k) H_m &= \left( H_m(-k) + \omega_m(k) \right) a_\lambda(k)  + e f_A(k) \cdot  v + e S \cdot f_B(k) ,
}
which holds for a.e. $k$ as  an identity  of  
  measurable functions.
Applying this to $\psi_m$ and using that $H_m \psi_m =  e_m  \psi_m $ we find  for  a.e. $k$
 \begin{align}\label{eq:pullident2}
((H_m( - k) - e_m + \omega_m(k)) a_\lambda(k) \psi_m = - (e f_A(k)  \cdot v + e S \cdot  f_B(k) )\psi_m.
\end{align}
This implies that $a_\lambda(k) \psi_m$ is in the domain of $H_m(-k)$ for a.e. $k$. Indeed,  the map
$$l: D_\mathcal{S} \to \C, \quad \eta \mapsto \inn{ a_\lambda(k) \psi_m ,   H_m(-k)   \eta  }$$ is bounded, since
in view of   \eqref{eq:pullident2} we can write
 $$l(\eta) =  \inn{   - (e f_A(k) \cdot  v + e S \cdot  f_B(k) )\psi_m + ( e_m - \omega_m(k)) a_\lambda(k) \psi_m , \eta } . $$  Now it follows that $a_\lambda(k) \psi_m \in D(H_m(-k))$, because
   $H_m(-k)$ is essentially self-adjoint
on  $D_\mathcal{S}$,   in view of  Theorem   \ref{hyp:op0}.
Hence the lemma follows  by applying $R_m(k)$ to \eqref{eq:pullident2}.
\end{proof}

\begin{lemma} \label{lem:estimate}
  Suppose there exists an $m_0 > 0$ such that   \eqref{eq:eineq} holds for all $m \in (0,m_0)$.
Then there exists a constant $C$ such 
that for all  $e \in \R$,  $m \in (0,m_0)$,  $k \in \R^3$, and $j=1,2,3$,
\begin{itemize}
\item[(a)] $ \| R_m(k) v_{j} \psi_m \| \leq C \omega_m^{-1/2}(k) ( 1 +
|k|^{1/2} ) , $
\item[(b)]
$ \|R_m(k) v_j|_{D(|P_{f}| ) \cap D(H_f^{1/2})}\| \leq C \omega_m(k)^{-1}  (1 + |k|)  . $
\end{itemize}
\end{lemma}

\begin{proof}
(a) We start with the product inequality
\begin{eqnarray} \label{eq:prod}
\| R_m(k) v_{j} \psi_m \| \leq  \| R_m(k)
( h_m - e_m )^{1/2} \| \| ( h_m - e_m )^{-1/2}  v_{j} \psi_m \| .
\end{eqnarray}
By Lemma \ref{lem:1} the second factor on the right hand side can be estimated using
$$
\| (h_m - e_m)^{-1/2}  v_{j} \psi_m \| \leq 1/\sqrt{2} \; .
$$
It remains to estimate the first factor in \eqref{eq:prod}.
First we use the trivial identity
\muu
{
h_m - e_m &= \frac{1}{2} (v-k)^2 + H_{f,m} - e_m + \frac{1}{2}k^2 + (v-k)k .
}
Estimating the last term using
$$
 (v-k)k \leq \frac{1}{2} |k| + \frac{1}{2} |k| (v-k)^2 ,
$$
we find
with
$
 \frac{1}{2} (v-k)^2 \leq H_m(-k)
$ that
\begin{align} \label{eq:hmemineq2}
h_m - e_m
&\leq (1 + | k |) (H_m(-k) + \omega_m(k) - e_m) + \frac{1}{2}(|k| +k^2) + | k | e_m .
\end{align}
Now  multiplying this inequality on both sides  with the self-adjoint operator $R_m(k)$ we obtain
\muu
{
R_m(k) (h_m - e_m) R_m(k) \leq  (1+ | k |) R_m(k) + \left( \frac{1}{2} ( | k |  + k^2) +  |  k | e_m \right) R_m(k)^2.
}
Using this, we estimate
\begin{align}
& \| R_m(k) (h_m - e_m)^{1/2}  \|^2 \\
&  \qquad \leq  \|  (h_m - e_m)^{1/2} R_m(k) \|^2 \nonumber \\
& \qquad = \sup_{\|\phi\| = 1} \inn{\phi,R_m(k)(h_m - e_m) R_m(k) \phi}  \nonumber   \\
& \qquad \leq (1 + | k |) \| R_m(k)\| +  \left( \frac{1}{2} ( |k | + k^2) + |k | e_m \right) \| R_m(k)\|^2.
\label{eq:last_term}
\end{align}
   Using \eqref{eq:boundonres}  and that  $e_m$ is bounded for  $0 \leq m \leq m_0$, we see that  \eqref{eq:last_term} inserted
    in \eqref{eq:prod}  implies the bound stated in (a). \\
(b)  Using that $v^2 \leq h_m$ we see from  \eqref{eq:hmemineq2} that
\begin{align*}
v^2
&\leq (1 + | k |) (H_m(-k) + \omega_m(k) - e_m) + \frac{1}{2}(|k| +k^2) + (1 + | k |) e_m .
\end{align*}
This implies
\begin{align}
\nonumber  & \left\|R_m(k) v_i \big|_{D(|P_{f}|) \cap D(H_f^{1/2})} \right\|^2 \\ \nonumber   &\qquad\qquad\leq (1 + | k |) \| R_m(k) \| +  \left( \frac{1}{2} (  |k | + k^2) + ( | k | +1) e_m \right) \| R_m(k)\|^2 \\
&\qquad\qquad\leq C (1+|k|^2) \omega_m(k)^{-2},
\end{align}
hence (b) follows.
\end{proof}

Estimating the expression in Lemma   \ref{lem:pull}  using Lemma  \ref{lem:estimate} (a)  we obtain  the
next lemma.

\begin{lemma} \label{lem:first}  Suppose Hypothesis A   holds.
Suppose there exists an $m_0 > 0$ such that   \eqref{eq:eineq} holds for all $m \in (0,m_0)$.
Then  there exists a finite
constant $C$
 such that for all $e \in \R$, $m\in (0, m_0)$ we have
$$
\| (a_\lambda(k) \psi_m ) \| \leq \frac{C | e \rho(k)|}{|k|} \quad   \ {\it for \
a.e.} \ \ k \; .
$$
\end{lemma}

We still need an estimate  involving   derivatives. To this end,  we shall henceforth make  an explicit choice
of the polarization vectors.
After a possible unitary
transformation on Fock space we can always achieve that the
polarization vectors are given by
\begin{eqnarray} \label{eq:cheps}
\varepsilon_{1}(k) = \frac{(k_2,-k_1, 0 ) }{\sqrt{k_1^2 + k_2^2}}
\quad \mathrm{and} \quad \varepsilon_{2}(k) = \frac{k}{|k|} \wedge
\varepsilon_1(k) \; .
\end{eqnarray}

\begin{lemma} \label{lem:second} Suppose Hypothesis A  holds.   Suppose there exists an $m_0 > 0$ such that   \eqref{eq:eineq} holds for all $m \in (0,m_0)$.
Then  there exists a finite constant $C$ such that for all $e \in \R$,  $m \in (0,m_0)$, and a.e. $k$ with $|k| < \Lambda$
$$
\| \nabla_k ( a_\lambda(k) \psi_m ) \| \leq \frac{C
| e \rho(k)|}{|k|\sqrt{k_1^2 + k_2^2}} \;  .
$$
\end{lemma}

\begin{proof}
We want to calculate  the derivative of the expression in Equation
\eqref{eq:pull}.  Calculating the derivative with respect  to  the operator norm topology,
we find by means of  the resolvent identity,  that
$$
\nabla_k R_m(k)  = -R_m(k)( (k-v) + \nabla_k \omega_m(k))R_m(k)   .
$$
Using this we can calculate the  derivative for   $0 < |k| < \Lambda$
\begin{align*}
& \nabla_k ( a_{\lambda}(k) \psi_m ) \\
 & =  -\frac{1}{2}\frac{e\rho(k)}{\sqrt{ 2 | k |} k^2} k   R_m(k)\left( -\varepsilon_\lambda(k) \cdot v  +  S \cdot (k \wedge \varepsilon_\lambda(k) \right) \psi_m   \\ 
& \quad -\frac{e\rho(k)}{\sqrt{2 \abs k}}  R_m(k) ( k - v + \nabla_k
\omega_m(k) ) R_m(k) \left( -\varepsilon_\lambda(k) \cdot  v  +  S \cdot (k \wedge \varepsilon_\lambda(k) \right)  ) \psi_m   \\ 
& \quad    + \frac{e\rho(k)}{\sqrt {2  | k|}}   R_m(k)\left( -\nabla_k \varepsilon_\lambda(k) \cdot v  +   \nabla_k (S \cdot(k \wedge \varepsilon_\lambda(k)) \right) \psi_m  .
\end{align*}
We now use that by Lemma \ref{lem:estimate},  there exists a
constant $C$ such that $\|\omega_m(k)R_m(k) v \psi_m \| \leq C
\omega_m(k)^{1/2} ( 1 + |k| )$. Using this together with   \eqref{eq:boundonres}  the first  and second
term are estimated from above by a finite constant times $|k|^{-2}$.
To estimate the third term we note that by the choice
\eqref{eq:cheps}, we have for $\lambda=1,2$,
$$
\left| \frac{\partial}{\partial k_l }  \varepsilon_\lambda(k) \right| \leq \frac{{\rm const.}}{\sqrt{
k_1^2 + k_2^2  }} \; , \quad l=1,2,3. 
$$
\end{proof}

\begin{lemma}[$y$-Bound] \label{thm:compact} Suppose Hypothesis A holds.   
 Suppose there exists an $m_0 > 0$ such that   \eqref{eq:eineq} holds for all $m \in (0,m_0)$, and let $e \in \R$. 
Then there exists  a constant $C$, and a $\delta
>0$ such that  for all $m \in (0,m_0)$ and all  $n \in \N$,
$$
\sum_{\lambda_1, \ldots, \lambda_n} \int   \sum_{i=1}^n n^{-1}
|y_i|^\delta \| (\widehat{\psi}_m)_{(n)}(\lambda_1,{y}_1, \ldots ,
\lambda_n , y_n) \|^2 \, d y_1  \ldots  dy_n  \leq C \; ,
$$
where $(\widehat{\psi}_m)_{(n)}$ denotes the Fourier transform of
the $n$-photon component of $\psi_m$.
\end{lemma}
\begin{proof} 
 We drop the subscript $m$. Thus by  $\widehat{\psi}_{(n)}$ we denote the
Fourier transform of $\psi_{(n)}$ in all its $n$-components. We
define the functions
\begin{eqnarray*}
\psi_{(n)}(k) &:& (\lambda,  k_1, \lambda_1, \ldots , k_{n-1}, \lambda_{n-1} ) \mapsto \psi_{(n)}( k, \lambda ,
k_1, \lambda_1 ,  \ldots
, k_{n-1}, \lambda_{n-1}  )  \\
\widehat{\psi}_{(n)}(y) &:& (\lambda, y_1, \lambda_1,  \ldots , y_{n-1} , \lambda_{n-1}  ) \mapsto
\widehat{\psi}_{(n)}(y , \lambda ,y_1, \lambda_1 , \ldots , y_{n-1} , \lambda_{n-1} ) \; .
\end{eqnarray*}

\vspace{0.5cm}

\noindent \underline{Step 1:} There exists a  $\delta >0$ and a
constant $C$ such that for all $a \in \R^3$,
$$
\int | 1 - e^{-iay} |^2 \| \widehat{\psi}_{(n)}(y) \|^2 dy \leq
\left\{
\begin{array}{ll} C |a|^{\delta} \quad & {\rm if} \  |a| < \frac{1}{2} \Lambda , \\
C \quad & {\rm if} \  |a| \geq \frac{1}{2}  \Lambda  .  \end{array}
 \right.
$$
The claim follows easily for $|a|\geq \frac{1}{2} \Lambda$, since $\psi$ is a
normalized state in Fock space and $|1-e^{-iay}| \leq 2$. Now lets
consider the case $|a| < \frac{1}{2}  \Lambda$.  By the Fourier transform, we have the
identity
\begin{align}
&  \int | 1 - e^{-iay} |^2 \| \widehat{\psi}_{(n)}(y) \|^2 dy  \nonumber \\
& \quad = \int \|
\psi_{(n)}(k + a ) - \psi_n(k) \|^2 dk \nonumber \\
& \quad  =  \int_{|k| < \Lambda - |a|}   \|
\psi_{(n)}(k + a ) - \psi_n(k) \|^2 dk   +   \int_{\Lambda - |a| \leq |k| }   \|
\psi_{(n)}(k + a ) - \psi_n(k) \|^2 dk .  \label{eq:yboundfourier}
\end{align}
To estimate the second integral  we use  Lemma \ref{lem:first}  and observe that the integrand vanishes for $\abs k > \Lambda + \abs a$,
\begin{align}
\int_{\Lambda - \abs a \leq \abs k }  \|
\psi_{(n)}(k + a ) - \psi_n(k) \|^2 dk  &\leq   {\rm const. }  \int_{\Lambda - \abs a \leq \abs k \leq \Lambda + \abs a}  \left( \frac{1}{\abs{k+a}^2} + \frac{1}{\abs{k}^2} \right)dk \nonumber \\
 &\leq   {\rm const.}   \int_{\Lambda - 2 \abs a \leq \abs k \leq \Lambda + 2 \abs a} \frac{1}{\abs{k}^2} dk \nonumber \\
 &\leq   {\rm const.}   \abs a .  \label{eq:intofpsi1}
\end{align}
Next we estimate the first integral  and
assume  $\abs k < \Lambda - \abs a$.  Using Lemma \ref{lem:second}, we find
\begin{align}
\| \psi_{(n)}( k + a ) - \psi_{(n)}(k) \|  & =  \left\| \int_0^1
\left(
\frac{d}{dt} \psi_{(n)}(k + t a) \right) dt  \right\|  \nonumber \\
& \leq |a| \int_0^1 \| \nabla_k \psi_{(n)}(k + t a) \| dt  \nonumber \\
& \leq  {\rm const.}   |a| \int_0^1 \frac{ \rho(k + t a )}{ |
k + t a | | \pi_3( k + t a ) | } dt \; , \label{eq:intofpsi}
\end{align}
where $\pi_3$ denotes the projection in $\R^3$ along the 3-axis and
const. denotes a finite constant independent of $n$. Let $\pi_a$
denote the projection in $\R^3$ along the vector $a$ and let
$\pi_{3,a}$ denote the projection in the $(1,2)$-plane along $\pi_3
a$  (with convention that $\pi_{3,a} = \pi_3$,  if $\pi_3 a = 0$). We find  from \eqref{eq:intofpsi}
\begin{eqnarray}  \label{eq:ybound1}
\| \psi_{(n)}(k + a ) - \psi_{(n)}(k) \| \leq {\rm const.}  \frac{
|a| }{|\pi_a(k) | | \pi_{3,a}(k) |} \; .
\end{eqnarray}
On the other hand using Lemma \ref{lem:first} we obtain
\begin{eqnarray}  \label{eq:ybound2}
\| \psi_{(n)}(k + a) - \psi_{(n)}(k) \| \leq {\rm const.}   \left(
\frac{\rho(k+ a)}{|k+a| }  + \frac{\rho(k)}{|k|}
\right) \; .
\end{eqnarray}
Introducing Inequalities \eqref{eq:ybound1} and \eqref{eq:ybound2}
into the first  integral of  \eqref{eq:yboundfourier}, we find for any $\theta$
with $0 \leq \theta \leq 1$,
\begin{eqnarray*}
 \lefteqn{ \int_{
\abs k < \Lambda - \abs a} \| \psi_{(n)}(k + a ) - \psi_n(k) \|^2 dk  } \\
&&= {\rm const.}  |a|^{2 \theta} \int_{|k| < \Lambda - |a|}  \frac{1}{|\pi_a(k) |^{2\theta}
|\pi_{3,a}(k) |^{2 \theta} } \left( \frac{\rho(k+a)}{|k+a|}
+ \frac{\rho(k)}{|k|} \right)^{2(1-\theta)} \, dk
\end{eqnarray*}
Now we use Young's inequality: $b c \leq b^p/p + c^q/q$, whenever
$p,q
> 1$ and $p^{-1} + q^{-1} = 1$; and the convexity of $x \mapsto x^{2(1-\theta)q}$ on $\R_+$, for
$0 < \theta < 1/2$. Thus for   $0 < \theta < 1/2$,
\begin{align}
&  \int_{  \abs k < \Lambda - \abs a} \| \psi_{(n)}(k + a ) - \psi_n(k) \|^2 dk
\nonumber \\
& \quad \leq |a|^{2 \theta} {\rm const.}   \int_{|k| \leq \Lambda  }
\Bigg( \frac{1}{|\pi_a(k)|^{4 \theta p}} +
\frac{1}{|\pi_{3,a}(k)|^{4 \theta p}} \nonumber \\
& \quad \hspace{3cm} + \left[ \frac{1}{|k+a|}
\right]^{2 (1-\theta) q}+ \left[ \frac{1}{|k|}
\right]^{2 (1-\theta) q} \Bigg) dk \; . \label{eq:intofpsi2}
\end{align}
For any $q$ with $1<q \leq 3/2$, we can choose $\theta > 0$ sufficiently
small such that the right hand side is finite.
Inserting  \eqref{eq:intofpsi1}  und  \eqref{eq:intofpsi2}  into \eqref{eq:yboundfourier}
we obtain the desired estimate.

 \vspace{0.5cm}

\noindent \underline{Step 2:} Step 1 implies the statement of the
Lemma.

From Step 1 we know that there exists a finite constant $C$ such
that
$$
\int \frac{|1 - e^{-iay} |^2 \| \widehat{\psi}_{(n)}(y)
\|^2}{|a|^{\delta/2} } dy \frac{da}{|a|^3} \leq C   \; .
$$
After interchanging the order of integration  and a change of
integration variables  $b = |y| a $, we find
\begin{eqnarray*}
C  \geq \int \| \widehat{\psi}_{(n)}(y) \|^2 \int \frac{ | 1 -
e^{-iay} |^2}{|a|^{\delta/2} } \frac{da}{|a|^3} dy = \int \|
\widehat{\psi}_{(n)}(y) \|^2 |y|^{\delta/2} \underbrace{ \int \frac{
| 1 - e^{-iby/|y|} |^2}{|b|^{\delta/2} } \frac{db}{|b|^3}  }_{=: \ c}  dy
\; ,
\end{eqnarray*}
where $c$ is nonzero and does not depend on $y$.

\end{proof}

\subsection{Existence of the  Ground State}

\vspace{0.5cm}

\noindent {\it Proof of Theorem \ref{thm:main1}}. Fix a positive
$m_0$ such that
for all  $m \in (0,m_0)$ the  energy inequality  \eqref{eq:eineq} holds.

\vspace{0.5cm}  \noindent \underline{Step 1:} All $\psi_m$, with
$m_0 \geq m
> 0$, lie in a compact subspace  of the reduce Hilbert space $\HH$.

\vspace{0.5cm}

Let $T$ be  the self-adjoint operator  associated to
the nonnegative and closed quadratic form $q$   in $\HH$
defined by
$$
q(\phi)  :=  \inn{ \phi , N \phi }  + \sum_{n=1}^\infty n^{-3} \inn{
\widehat{\phi}_{(n)} , \sum_{i=1}^n |y_i|^{\delta}
\widehat{\phi}_{(n)}  } + \inn{ \phi , H_f   \phi }  \; ,
$$
for all $\phi \in D(q)$, the natural  form domain of $q$. We choose $\delta
> 0$ such that Lemma \ref{thm:compact} holds. By this and Lemma
\ref{lem:first}  and Proposition  \ref{eq:propenegyconv2}, there exists a finite $C$ such that for all $m$
with $0 < m <m_0$,
$$
\psi_m \in K := \{ \phi \in D( q) : \| \phi \|  \leq 1,  q(\phi)
\leq C \} \; .
$$
 The set $K$ is a compact subset of $\HH$, provided $T$ has compact resolvent
\cite[Theorem XIII.64]{ressim:ana}. Hence it remains to show that $T$ has compact
resolvent. The operator $T$ preserves the $n$-photon sectors. Let
$T_n$ denote the restriction of $T$ to the $n$-photon sector. From
Rellich's criterion  \cite[Theorem XIII.65]{ressim:ana}   it follows that $T_n$ has  compact resolvent.
Therefore $\mu_l(T_n) \to \infty$ as $l$ tends to infinity, where
$\mu_l$ denotes the $l$-th eigenvalue obtained by the min-max
principle. Moreover since $\mu_l(T_n) \geq n$ for all $l, n$, it
follows that $\mu_l(T) \to \infty$ as $l \to \infty$. Hence $T$ has
a compact resolvent.

\vspace{0.5cm}

\noindent \underline{Step 2:} There exists a nonzero vector $\psi_0$
such that $\inn{ \psi_0 , H(0) \psi_0 }  = \inf \sigma(H(0))$.

\vspace{0.5cm}

Here we use the argument outlined at the beginning of this section.
By Step 1, we know that  all $\psi_m$, with $m_0 \geq m
> 0$, lie in a compact subspace  of $\HH$. It follows that
there exists a subsequence $(\psi_{m_i})_{i \in \N}$, with $m_i \to
0$ as $i\to \infty$, which converges strongly to a normalized vector
$\psi_0$. By lower semicontinuity of non-negative quadratic forms we
see from Proposition  \ref{eq:propenegyconv2}  that
\begin{eqnarray*}
\lefteqn{ \inn{  \psi_0 , ( H(0)  - E(0) ) \psi_0 }  } \\
&& \leq \liminf_{i \to \infty} \inn{  \psi_{m_i} , ( H(0) - E(0)
\psi_{m_i} } = 0 \; .
\end{eqnarray*}
This shows Step 2.
 \qed

\section*{Acknowledgements}

D. H. wants to thank Ira Herbst for valuable discussions. In particular  the proof of a  key
idea, Lemma  \ref{thm:compact},  is from   Ira Herbst.

\appendix

\section{Self-adjointness }

 \label{appessself}

In this section we prove Theorem  \ref{hyp:op0}. The proof is based on an inequality similar to   \cite{HH08}.
In contrast to the proof given in  that paper, where the domain of self-adjointness  is determined by means  of  quadratic  forms, we use the following abstract
proposition, which can be derived     from a theorem of W\"ust, similar to \cite{omatte}.

\begin{proposition}\label{prop:wuest}
Let $T$ be a self-adjoint operator on a Hilbert space and  let $T^{(n)}$, $n=1,2$ be  symmetric and $T$-bounded
operators.
 For $\kappa \in \C$ let    $T(\kappa) = T+ \kappa T^{(1)} + \kappa^2 T^{(2)}$  be the  operator with  domain $D(T)$.
If $T(\kappa)$ is    closed for all $\kappa  \in [0,t]$, then $T(t)$ is self-adjoint.
\end{proposition}
\begin{proof}
Let $Z = \{ \kappa \in \C : T(\kappa) \text{  is closed  }  \}$. We claim that  $Z$ is open.
If $\kappa_0 \in Z$, then $T(\kappa_0)$ is closed and $D(T(\kappa_0)) = D(T)$.
The operators $T^{(j)}$ are closable operators such that $D(T(\kappa_0))  = D(T) \subset  D(T^{(j)})$.
Then by the closed graph theorem $T^{(j)}$ are also $T(\kappa_0)$ bounded \cite[Theorem 5.9]{weid80}.
It follows that $T(\kappa)$ is closed for $\kappa$ close to $\kappa_0$
\cite[Theorem 5.5]{weid80}. Thus we have
shown that  $Z$ is open.  It follows that $\kappa \mapsto T(\kappa)$ is on $Z$
a holomorphic family of type (A) \cite{kato}.
Since $T^{(j)}$ are symmetric  and $T$-bounded, we have   $T(\overline{\kappa}) \subset T(\kappa)^* $ for all $\kappa \in \C$. Let $Z_0$ denote the connected component of $Z \cap \overline{Z}$ containing $0$. Since $T(0) =T$ is self-adjoint, and hence $0 \in Z_0$, it follows  from a
Theorem of  W\"ust,
\cite[Theorem 1]{wuest}, that $T(\overline{\kappa}) = T(\kappa)^*$ for all $\kappa \in Z_0$.
Since   $[0,t] \subset Z_0$, it follows that $T(\kappa)$ is self-adjoint for all $\kappa \in [0,t]$.
\end{proof}

We apply the above Proposition  to $T(e) = T +  e  T^{(1)} + e^2 T^{(2)}$, where  $T = \frac{1}{2}(\xi  - P_f)^2 + H_{f,m}$, with natural domain, $T^{(1)} = \frac{1}{2} ( P_f -\xi ) \cdot  A  +\frac{1}{2}  A \cdot ( P_f -\xi )  +   S \cdot B$, and $T^{(2)} = \frac{1}{2}   A^2 $.
First we note the following.  Since $A$ is divergence free  we have $A \cdot  P_f = P_f  \cdot A$, and so    $T^{(1)}  =   A  \cdot P_f - \xi \cdot A + S \cdot B$.
Let  	$T_0 = \frac{1}{2} P_f^2 + H_{f,m}$, with natural domain. Since $T$ and $T_0$ are non-negative
multiplication operators, it is easy to see that the domains of  $T_0$ and  $T$ agree, and that the operators  are mutually bounded.
Thus,  as $T(e) = H_m(\xi)$,  Theorem  \ref{hyp:op0} will follow as a consequence of  Proposition \ref{prop:wuest}   and the following two lemmas.

\begin{lemma}  \label{eq:last1} Let $m \geq 0$, $\xi \in \R^3$,  and $\rho \in L^2(\R^3;(|k|+\omega_m(k)^{-1} |k|^{-1} ) dk)$. Then $T^{(1)}$ and $T^{(2)}$ are
 $T$-bounded.
\end{lemma}
\begin{proof} Standard estimates, see for example \cite{ressim:fou},  show  that $A^2$ is $H_{f,m}$-bounded and that the components of  $A$ and $B$ are $H_{f,m}^{1/2}$-bounded.
It follows that  for some $C$  we have  $\| A_l P_f  \psi \| \leq C ( \| H_{f,m}^{1/2}  P_f \psi \| + \| P_f \psi \| ) $. Using   $\| H_{f,m}^{1/2}  P_f \psi \|^2 \leq   \| ( P_f^2 + H_{f,m} ) \psi \|^2  $ and
collecting  estimates shows the claim.
\end{proof}

\begin{lemma}  \label{eq:last2} Let  $m \geq 0$ and $\rho \in L^2(\R^3;(|k| + \omega_m(k)^{-1}  |k|^{-1}) dk)$ . Then for any $e \in \R$ and $\xi \in \R^3$
there exist  constants $C_1, C_2$ such that for all $\varphi  \in D(T)$,
\begin{eqnarray} \label{eq:2nd}
 \|  T  \varphi  \|^2 \leq  C_1
 \|  T(e )\varphi \|^2  + C_2 \| \varphi \|^2 \; .
\end{eqnarray}
\end{lemma}

\begin{proof}  Let $e \in \R$ and $m \geq 0$ be fixed. For notational compactness we set  $Q:= P_f - \xi$ and $F := e A$.

\vspace{0.4cm}

 \noindent \underline{Step 1:} There exist constants $c_1, c_2, c_3 $ such
that
$$
\| Q^2 \varphi \|^2 \leq c_1 \|   (Q + F )^2  \varphi \|^2   +
c_2 \| H_{f,m}  \varphi \|^2   + c_3 \| \varphi \|^2 \ ,
\quad \forall \varphi \in D(T)  \; .
$$

In the following we denote by $C$ a constant which may change from line to line.
We have
\begin{align}
\|  Q^2 \varphi \|^2
&=  \|  (  ( Q + F  )^2  - 2 F   \cdot ( Q + F  )  +
F ^2 ) \varphi \|^2 \nonumber \\
&\leq 3 \| ( Q + F  )^2 \varphi \|^2  + 12 \|  F   \cdot ( Q + F
) \varphi  \|^2 + 3 \| F ^2 \varphi \|^2 \; . \label{eq:today33}
\end{align}
The second  term in \eqref{eq:today33} is estimated as follows
\begin{eqnarray*}
  \| F  \cdot  ( Q + F  ) \varphi \|^2   \leq C \sum_j \| F _j ( Q_j + F_j) \varphi \|^2
\leq C \sum_j \| ( H_{f,m} + 1 )^{1/2} ( Q_j + F_j ) \varphi \|^2 \; .
  \end{eqnarray*}
Further, using a commutator
\begin{align*}
 \sum_j \| ( &H_{f,m} + 1 )^{1/2} ( Q_j + F_j ) \varphi \|^2
\\ &=
\sum_j \inn{  (Q_j + F_j) \varphi , ( H_{f,m} + 1 ) (Q_j + F_j) \varphi } \\
&=  \sum_j  \left\langle ( Q_j + F _j)^2 \varphi , ( H_{f,m} + 1 ) \varphi \right\rangle +
\left\langle ( Q + F )_j \varphi , [ H_{f,m} , F _j ] \varphi  \right\rangle  \\
& \leq  C (  \| (Q + F )^2 \varphi \|^2 +    \sum_j \| (Q_j + F_j) \varphi \|^2     +
  \| ( H_{f,m} + 1 ) \varphi \|^2 ) \\
& \leq   C (  \| (Q + F )^2 \varphi \|^2 +      \| ( H_{f,m} + 1 ) \varphi \|^2 ) \; .
  \end{align*}
Collecting the above estimates yields Step 1.

\vspace{0.4cm}

 \noindent \underline{Step 2:} There exists a constant $C$ such
that
\begin{eqnarray*}
\|\sfrac{1}{2} (Q + F )^2 \varphi \|^2  + \| H_{f,m} \varphi \|^2  \leq
 \| ( \sfrac{1}{2} ( Q + F )^2 + H_{f,m} )\varphi \|^2 + C \| \varphi \|^2 \ ,
\quad \forall \varphi \in D(T)  \; .
\end{eqnarray*}

\vspace{0.4cm}

Calculating a double commutator,
we see that
\begin{eqnarray}
 \lefteqn{ \sfrac{1}{2}   \inn{  H_{f,m} \varphi , ( Q + F )^2 \varphi }  +
\inn{ ( Q + F )^2 \varphi , H_{f,m} \varphi)   } }  \nonumber \\
&=&
\sum_j    \inn{  (Q_j +F _j) \varphi ,  H_{f,m} ( Q_j +F _j) \varphi}
   +  \sfrac{1}{2} \sum_j  \inn{ \varphi, [ F _j, [
F _j , H_{f,m} ]]  \varphi}   \nonumber  \\
&\geq&  - b \| \varphi \|^2 \; ,   \label{eq:last11}
\end{eqnarray}
 for some $b$. Step 2 follows by adding Inequality   \eqref{eq:last11} to the left hand side and completing the square.

\vspace{0.4cm}

 \noindent \underline{Step 3:} Inequality \eqref{eq:2nd} holds.

\vspace{0.4cm}
First observe that
\begin{eqnarray} \label{eq:2nd2}
\| ( \sfrac{1}{2} Q^2 + H_{f,m} ) \varphi \|^2   \leq \frac{1}{2}  \| Q^2 \varphi \|^2 +
2 \|  H_{f,m} \varphi \|^2    \ ,
\quad \forall \varphi \in D(T)  \; .
\end{eqnarray}
Inserting  on the right hand side of   \eqref{eq:2nd2} first the inequality of Step 1 and then the inequality of Step 2, we find for some constant $C$
\begin{eqnarray} \label{eq:2nd22}
\| ( \sfrac{1}{2} Q^2 + H_{f,m} ) \varphi \|   \leq  C \| ( \frac{1}{2} ( Q+ F)^2  + H_{f,m} ) \varphi \| + \| \varphi \| )    \ ,
\quad \forall \varphi \in D(T)  \; .
\end{eqnarray}
By standard  estimates  $S \cdot B$ is infinitesimally bounded with respect to $H_{f,m}$.
It follows by   \eqref{eq:2nd22} that it is also  infinitesimally bounded with respect to  $\frac{1}{2} ( Q+ F)^2  + H_{f,m} $. Thus   $\frac{1}{2} ( Q+ F)^2  + H_{f,m} $
is $T(e)$-bounded. Now   \eqref{eq:2nd} follows  from  \eqref{eq:2nd22}.
\end{proof}

\section{Ground State for positive Photon Mass} \label{appexreg}

 In this section we provide a  proof of Theorem \ref{hyp:op1}.  It follows closely the proofs  given in   \cite{GriLieLos:01,LMS06}. Theorem \ref{hyp:op1} will follow directly from Propositions   \ref{prop:essspec}  and  \ref{prop:deltpos}, below.
For $\xi \in \R^3$  and $m \geq 0$ we define
$$
\Delta_m(\xi) = \inf_{k \in \R^3} \left\{   E_m(\xi-k) - E_m(\xi) + \omega_m(k) \right\} .
$$

\begin{proposition} \label{prop:essspec}  Let $\rho \in L^2(\R^3;(|k| + |k|^{-1}) dk)$ and $m > 0$. Then 
for all $\xi \in \R^3$ we have
$$
{\rm inf} \sigma_{\rm ess}(H_m(\xi)) \geq E_m(\xi) + \Delta_m(\xi) .
$$
\end{proposition}
To prove the proposition we need the following notation.
Let $\hh_1$ and $\hh_2$ be two Hilbert spaces. If  $A : \hh_1 \to \hh_2$ is a partial isometry,
we define $\Gamma(A)$ to be the linear operator $\FF(\hh_1) \to \FF(h_2)$ which equals
$\bigotimes_{k=1}^n A$ when restricted to $\hh_1^{(n)}$,  $ n \geq 1$, and which equals to the identity on $\hh_1^{(0)}$.  The following two lemmas are straightforward to verify.
\begin{lemma} \label{eq:lemmgmma}  Let $\hh_1$ and $\hh_2$ be two Hilbert spaces,
and  $A : \hh_1 \to \hh_2$ a partial  isometry. Then  $\Gamma(A)^* = \Gamma(A^*)$, and
$$
\Gamma(A) a^*(f) = a^*(A f) \Gamma(A) . 
$$
If $A$ is an isometry, then so is $\Gamma(A)$ and
$
\Gamma(A) a(f) = a(A f ) \Gamma(A) .
$
\end{lemma}

\begin{lemma} \label{eq:lemmsumprod}  Let $\hh_1$ and $\hh_2$ be two Hilbert spaces. Then  there exists a unique bounded linear map  $U : \FF(\hh_1 \oplus \hh_2) \to \FF(\hh_1) \otimes \FF(\hh_2)$
 such that
\begin{align*}
U \Omega  & = \Omega \otimes \Omega , \\
U (a^* (h_1,h_2))   & =  ( a^*(h_1) \otimes \one + \one \otimes a^*(h_2) ) U  ,  \quad \forall (h_1,h_2) \in \hh_1 \oplus \hh_2 .
\end{align*}
It follows that  $U$ is unitary.
\end{lemma}

Let $\chi_1,\chi_2 \in \C^\infty(\R^3;[0,1])$ with $\chi_1^2 + \chi_2^2 = 1$ and
$\chi_1(x) = 1$, if $|x| < 1$, and $\chi(x) = 0$, if $|x| > 2$.   We define  $j_l  = \chi_l(-i \nabla_k /L)$ for $l=1,2$ and   $L > 0$.
Define $$j : \hh \to  \hh \oplus \hh , \quad f \mapsto (j_1 f, j_2 f ), $$
with $\hh$ given by   \eqref{eq:defofhilconc}. Henceforth we denote  the identity map of $\hh$ by $1$.
One readily verifies that $j$ is an isometry. Hence  $j^*$ is a partial
isometry and $j^* j = 1$.  Explicitly, one finds     $j^*(f_1,f_2) = j_1 f_1 + j_2 f_2$ for $f_j \in \hh$, $j=1,2$.
Henceforth, let $U$ denote the  isometry as in Lemma  \ref{eq:lemmsumprod} with $\hh_1 = \hh_2 = \hh$, and
let  $J = U  \Gamma(j)$. Recall that $N = d \Gamma(1)$.
 In the following let $a^\#$ stand for $a^*$ or $a$.

\begin{lemma} \label{lem:jjest} The following holds.
\begin{itemize}
\item[(a)]  $J^* J = \one$.
\item[(b)]   $J a(f)^\# = \left\{ a(j_1 f)^\# \otimes \one + \one \otimes  a(j_2  f)^\# \right\}  J .  $
\item[(c)] We have
\begin{align*}
J^* (a^*(f_1) \otimes \one + \one \otimes a^*(f_2) )  & = a^*(j_1 f_1+j_2 f_2) J^* , \\
 (a(f_1) \otimes \one + \one \otimes a(f_2) )  J  & =  J a(j_1  f_1+j_2  f_2)  .
\end{align*}
\end{itemize}
\end{lemma}
\begin{proof} The Lemma follows directly from the definition of $J$  and the properties of Lemmas \ref{eq:lemmgmma} and \ref{eq:lemmsumprod}.
\end{proof}

\begin{lemma}  \label{lem:jcommfield} Let $f \in \hh$.
\begin{itemize}
\item[(a)]   For $\psi \in D(N^{1/2}) $ we have
 \begin{align}
 \|(  J a(f) -  ( a(f)  \otimes \one )  J ) \psi \|   & \leq    \| ( 1-  j_1 )  f  \|  \| N^{1/2}  \psi \|    ,  \label{fieldloc1} \\
 \|(  J a^*(f) -  ( a^*(f)  \otimes \one )  J ) \psi \|  & \leq   ( \| (1 -  j_1   ) f  \| + \| j_2 f \| ) \| (N+1)^{1/2}  \psi \|\label{fieldloc2}    .
\end{align}
\item[(b)]  We have 
\begin{align*}
J a^\#(f) -  ( a^\#(f) \otimes \one ) J = ( a^\#(( j_1 - 1 ) f ) \otimes \one + \one \otimes a^\#(j_2 f) ) J   .
\end{align*}
\end{itemize}
\end{lemma}
\begin{proof}
Part  (b) follows directly from   Lemma \ref{lem:jjest} (b).
To show  \eqref{fieldloc1}, we insert  in (b) the identity from Lemma \ref{lem:jjest} (c) and find
\begin{align*}
J a(f) -  ( a(f) \otimes \one ) J =  J a(j_1 ( j_1 - 1 ) f + j_2^2 f  ) = J a((1-j_1) f )  .
\end{align*}
Now the inequality follows from standard estimates.  To show  \eqref{fieldloc2}
 we again use   (b),
\begin{align*}
J a^*(f) -  ( a^*(f) \otimes \one ) J =    ( a^*(( j_1 - 1 ) f ) \otimes \one + \one \otimes a^*(j_2 f) ) J  .
\end{align*}
To estimate the second term on the right hand side we first use the canonical commutation relations 
and then   Lemma  \ref{lem:jjest} (c)  and find
\begin{align*}
J^* ( \one \otimes a(j_2 f)   a^*(j_2 f) ) J  & =   \| j_2 f \|^2  J^* J +    J^* ( \one \otimes a^*(j_2 f )  a(j_2 f) ) J \\
& =  \| j_2 f \|^2   +     a^*(j_2^2 f )  J^* J  a(j_2^2 f) \\
& \leq \| j_2 f \|^2  +  \| j_2^2 f \|^2 N \\
&  \leq \| j_2 f \|^2 (1 + N)  .
\end{align*}
To estimate the first  term on the right hand side we find similarly
\begin{align*}
& J^* (  a((j_1-1) f) )  a^*((j_1-1) f) \otimes \one  ) J  \\
& \quad =   \| (j_1-1) f \|^2  J^* J +    J^*  a^*((j_1-1) f )  a((j_1-1) f) \otimes \one   J
\\
&   \quad =  \| j_2 f \|^2   +     a^*(j_1 (j_1-1) f )  J^* J  a(j_1(j_1 - 1)  f) \\
& \quad \leq \| ( j_1-1) f \|^2  +  \| j_1(j_1-1) f \|^2 N  \leq  \| ( j_1-1) f \|^2  (1  +   N ) .
\end{align*}
Collecting estimates shows \eqref{fieldloc2}.
\end{proof}

\begin{lemma} \label{lem:jcommgamma}
 Let $h$ be a selfadjoint operator in $\hh$. Suppose there exists a dense subspace $\mathcal{D} \subset D(h)$ such that $j_l(\mathcal{D}) \subset D(h)$  for  $l=1,2$.
\begin{itemize}
\item[(a)]   Suppose  that $[j_l,h]$ is bounded for  $l=1,2$. Then
for $\psi \in D(d \Gamma(h)) \cap D(N) $ we have
 $$ \|(  J d \Gamma(h) -  ( d \Gamma(h) \otimes \one + \one \otimes d \Gamma(h)  ) J ) \psi \|  \leq   ( \| [j_1,h ] \|^2 + \| [ j_2 , h ] \|^2 )^{1/2} \| N \psi \|   . $$
\item[(b)] Let $(e_l)_{l \in \N}$ be an orthonormal basis of $\hh$ which lies in  $\mathcal{D}$.    Then
\begin{align*}
& J d \Gamma(h)  - ( d \Gamma(h) \otimes \one + \one \otimes d \Gamma(h)  ) J \\
&  \quad = \sum_l  \left(  a^*([j_1,h]e_l) \otimes \one +  \one \otimes a^*([j_2,h]e_l)   \right)   J a(e_l)  ,
\end{align*}
where the expression on the right hand side is understood in the weak sense.
\end{itemize}
\end{lemma}
\begin{proof} In the proof let $\inn{ \cdot , \cdot }$ denote the scalar product in $\hh$.
First we show (b).
We find using Lemma \ref{lem:jjest} (b)
\begin{align*}
J d \Gamma(h) &  = J \sum_{l,k} \inn{ e_l, h e_k} a^*(e_l) a(e_k)  \\
&  =  \sum_{l,k} \inn{ e_l, h e_k} ( a^*(j_1 e_l) \otimes \one + \one \otimes a^*(j_2 e_l))    J a(e_k)  \\
&  =  \sum_{l,k} (\inn{ e_l, j_1 h e_k} a^*(e_l) \otimes \one + \inn{ e_l, j_2 h e_k} \one \otimes a^*( e_l))    J a(e_k)  .
\end{align*}
Using Lemma \ref{lem:jjest} (c) we find
\begin{align*}
( d \Gamma(h) &\otimes \one + \one \otimes d \Gamma(h)  ) J \\ &  =
\sum_{l,k} \inn{ e_l, h e_k} ( a^*(e_l) a(e_k)   \otimes \one + \one \otimes a^*(e_l) a(e_k)    ) J  \\
&  =  \sum_{l,k} \inn{ e_l, h e_k} ( ( a^*( e_l) \otimes \one  ) J a(j_1 e_k)
+  ( \one \otimes a^*( e_l) )    J a(j_2 e_k) )   \\
&  =  \sum_{l,k} ( \inn{ e_l, h j_1e_k}  ( a^*( e_l) \otimes \one  ) J a( e_k)
+  \inn{ e_l, h j_2e_k} ( \one \otimes a^*( e_l) )    J a( e_k) )   .
\end{align*}
Taking the difference (b) follows.
 Now (a) follows from (b) using standard estimates.
For example using Lemmas \ref{eq:lemmgmma}  and
\ref{eq:lemmsumprod}  we find
\begin{align*}
& \text{Right  hand side of  (b)} \\
& \quad =  U \sum_l a^*(([j_1,h] e_l , [j_2,h] e_l) a(j_1 e_l, j_2 e_l ) \Gamma(j)   \\
& \quad =  U d \Gamma(A)  \Gamma(j)  ,
\end{align*}
where we defined the following  operator on $\hh \oplus \hh$
$$
A = \left( \begin{array}{cc}   [j_1,h] j_1    &    [j_1,h] j_2  \\  {[} j_{2},h {]} j_{1} & {[} j_{2},h {]} j_{2}  \end{array} \right) .
$$
Now the bound follows since the operator preserves the $n$-particle sector and satisfies the following estimate
\begin{align*}
& \| d \Gamma(A) \Gamma(j) |_{\hh^{(n)}}\| \\
 & \quad  \leq   \| ( A j  ) \otimes j \otimes \cdots \otimes  j \|   + \|  j \otimes ( A j  ) \otimes \cdots \otimes j \|  + \cdots +  \| j  \otimes \cdots \otimes j  \otimes (  A j  ) \|   \\
& \quad \leq   n \| A j \| \| j \|^{n-1} =  n \| A j \| = n   ( \| [j_1 , h ] \|^2 + \| [j_2 , h ] \|^2)^{1/2} .
\end{align*}
\end{proof}
Recall that $\HH = \C^{2s+1} \otimes \FF$.
We consider the map
\begin{align*}
\one \otimes J : \C^{2s+1} \otimes \FF \to \C^{2s+1} \otimes  ( \FF \otimes \FF)
\end{align*}
from $\HH $ to $\HH \otimes \FF$.  By abuse of notation we shall henceforth denote this map again by $J$.
We introduce  the operator
\begin{align*}
& \widetilde{H}_m(\xi)  := \frac{1}{2} ( \xi - P_f \otimes \one - \one \otimes P_f - e A \otimes \one )^2 +  e  S \cdot B \otimes \one + H_{f,m} \otimes \one + \one \otimes H_{f,m}
\end{align*}
on $\HH \otimes \FF$, with domain given by the natural domain of $\tilde{H}_m(\xi)|_{e = 0 }$.
\begin{lemma} \label{lem:exgsregmain}  Let  $m  >  0$ and $\rho \in L^2(\R^3;(|k| + |k|^{-1}) dk)$.
Then the following  holds.
\begin{itemize}
\item[(a)] For $\varphi \in D(H_m(\xi))$ we have
$$
| \inn{ \varphi, H_m(\xi) \varphi } -  \inn{ J \varphi ,  \widetilde{H}_m(\xi) J \varphi } | \leq o(L^0) ( \| H_m(\xi) \varphi \|^2 + \| \varphi \|^2 )  \quad ( L \to \infty) ,
$$
where $o(L^0)$  does not depend on $\varphi$.
\item[(b)]  For $\varphi \in D(\tilde{H}_m(\xi))$ we have
$$
\inn{  \varphi , \widetilde{H}_m(\xi) \varphi } \geq \inn{ \varphi, \left\{ E_m(\xi) +  \Delta_m(\xi) (\one - P_{\Omega,2}) \right\}  \varphi } ,
$$
where $P_{\Omega,2}$ denotes  the orthogonal projection in $\HH \otimes \FF$ onto $\HH \otimes \Omega$.
\end{itemize}
\end{lemma}
\begin{proof}
(a)  Defining  the operators
$$
Q = J( \xi  + v ) - (\xi - P_f \otimes \one -  \one \otimes  P_f  +   e A  \otimes \one ) J
$$
we can write
\begin{align*}
H_m(\xi) - J^* \widetilde{H}_m(\xi) J &  = \frac{1}{2} \left\{  (\xi + v)  J^* Q + Q^* J (\xi + v )   -   Q^* Q  
\right\}  \\
& \quad + e  S \cdot (  B  - J^*  ( B \otimes \one  ) J  ) \\
& \quad + H_{f,m }- J^*  (H_{f,m} \otimes \one + \one \otimes H_{f,m} ) J .
\end{align*}
Now using that $J$ is an isometry, it follows 
from Lemma  \ref{lem:jcommgamma}  (a)  (choosing for example $\mathcal{D} = C_c(\Z_2 \times \R^3)$)
 that  for $\varphi \in D(N)$ we have, 
$$
\| H_{f,m }- J^*  (H_{f,m} \otimes \one + \one \otimes H_{f,m} ) J ) \varphi \|  \leq  ( \| [ j_1,\omega_m] \| + \| [j_2,\omega_m ] \|)   \| N \varphi \| .
$$
Using again that  $J$ is an isometry
 it follows from Lemma  \ref{lem:jcommfield}  (a)   that  for $\varphi \in D(N^{1/2})$ we have, 
  recalling  the notation   \eqref{eq:defoffieldfunc},
$$
\| (  S \cdot (  B  - J^*  ( B \otimes \one  ) J  )   ) \varphi \|  \leq  ( \|( j_1-1) f_B   \| + \| j_2 f_B \|)  \| (N+1)^{1/2}  \varphi \| .
$$
Analogously, we find  from  Lemmas   \ref{lem:jcommgamma}  and   \ref{lem:jcommfield}  for $\varphi \in D(N)$
$$
\| Q \varphi \| \leq ( \| [ j_1,\nu] \| + \| [j_2, \nu  ] \|)   \| N \varphi \| + |e|   ( \|( j_1-1) f_A   \| + \| j_2 f_A \|)  \| (N+1)^{1/2}  \varphi \|  ,
$$
where  $\nu : \R^3 \to \R^3$  with  $\nu(k)=k$.
Now we will use  that for $m > 0$
$$
\| N \varphi \| \leq C (  \| H_m(\xi) \varphi \| + \| \varphi \|  ) ,
$$
and that 
\begin{align*}
 \|( j_1-1) f_B   \| ,   \| j_2 f_B \| ,  \|( j_1-1) f_A   \| , \| j_2 f_A \| = o(L^0) , 
\end{align*}
which follows by dominated convergence in Fourier space.  Furthermore, we note 
\begin{align*}
 & \| [ j_1,\nu] \| ,  \| [j_2, \nu  ] \|  = O(L^{-1}) ,  \\
&  \| [ j_1,\omega_m] \| , \| [j_2,\omega_m ] \| = O(L^{-1})  ,
\end{align*}
where the first line is easy to see and   the second line can be seen as  follows. Let $\chi_{1,L} = \chi_1(\cdot/L)$.   For normalized $\varphi_1, \varphi_2$ we find using the usual notation and  convention
for the Fourier transform
\begin{align*}
| \inn{ \varphi_1 , [j_1 , \omega_m ] \varphi_2 } | & = (2\pi)^{-3/2} \left| \int_{\R^3 \times \R^3}   \overline{\varphi_1(p)} \widehat{\chi_{1,L}}(p-q) ( \sqrt{ q^2 + m^2} - \sqrt{ p^2 + m^2 } ) \varphi_2(q) dp dq \right| \\
  & \leq  (2\pi)^{-3/2} \int_{\R^3 \times \R^3}  |\varphi_1(p) | \sum_{s=1}^3  | \widehat{\chi_{1,L}}(p-q) | |p_s - q_s| |\varphi_2(q)| dp dq   \\
  & =   (2\pi)^{-3/2} \int_{\R^3 \times \R^3}  |\varphi_1(p) |
\sum_{s=1}^3  | \widehat{ \partial_s \chi_{1,L}} (p-q) | |\varphi_2(q)| dp dq   \\
& \leq \sum_{s=1}^3  \| (| \widehat{ \partial_s \chi_{1,L}} |)^{\vee} \|_\infty = L^{-1} \sum_{s=1}^3  \|  (|\widehat{\partial_s \chi_1} |)^\vee \|_\infty .
\end{align*}
To treat  the term involving  $j_2$ we apply a similar  estimate  to  the function $1-\chi_2$.
Inserting the above estimates  Part   (a) now follows. \\
To show (b) we use the canonical isomorphism
$$
\HH \otimes \FF \cong      \bigoplus_{n \in \N_0} L^2_{\rm sym}(( \Z_2 \times \R^{3} )^n ; \HH )  ,
$$
where the summand with $n=0$ is by convention $\HH$.
With respect to this fiber decomposition the Hamiltonian $\widetilde{H}_m(\xi)$ fibrates and we have for $n \in \N_0$,
\begin{align} \label{eq:actofhamtilde}
(\widetilde{H}_m(\xi) \psi)_{(n)}(\lambda_1,k_1,\cdots,\lambda_n,k_n) & = \left( H_m(\xi - \sum_{j=1}^n k_j) + \sum_{j=1}^n \omega_m(k_j) \right) \psi_{(n)}(\lambda_1,k_1,\cdots,\lambda_n,k_n)  .
\end{align}
This yields the expectation
\begin{align*}
\inn{ \psi , \widetilde{H}_m(\xi) \psi } & = \sum_{n=0}^\infty  \inn{ \psi_{(n)} ,  ( \widetilde{H}_m(\xi)  \psi)_{(n)} } .
\end{align*}
To calculate the summands on the right hand side  we   again use   \eqref{eq:actofhamtilde}
and find for  $n \geq 1$
\begin{align*}
&  \inn{ \psi_{(n)} ,  ( \widetilde{H}_m(\xi)  \psi)_{(n)} }
 \geq  ( E_m(\xi) + \Delta_m(\xi) ) \| \psi_{(n)} \|^2 ,
\end{align*}
where  we employed  the following operator inequality. Using     $\omega_m(p) + \omega_m(q) \geq \omega_m(p+q)$ we find
\begin{align*}
  H_m\left(\xi - \sum_{j=1}^n k_j\right) + \sum_{j=1}^n \omega_m(k_j)   & \geq   E_m\left(\xi - \sum_{j=1}^n k_j\right) + \sum_{j=1}^n \omega_m(k_j) \\
 & \geq   E_m\left(\xi - \sum_{j=1}^n k_j\right) +  \omega_m\left( \sum_{j=1}^n k_j\right)  \\
& \geq E_m(\xi) + \Delta_m(\xi) .
\end{align*}
The above now implies
\begin{align*}
\inn{ \psi , \widetilde{H}_m(\xi) \psi } & \geq  E_m(\xi) \| \psi_{(0)} \|^2  +  (  E_m(\xi) + \Delta_m(\xi) )  \sum_{n=1}^\infty \| \psi_{(n)}\|^2  .
\end{align*}
Thus we have shown  (b).
\end{proof}
\begin{proof}[Proof of Proposition  \ref{prop:essspec}]
From Lemma \ref{lem:exgsregmain}  we find   with $\| \varphi \|_{H_m(\xi)} :=  ( \| H_m(\xi) \varphi \|^2 + \| \varphi \|^2)^{1/2}$
\begin{align} \label{eq:ineqforinfess}
\inn{ \varphi, H_m(\xi) \varphi } \geq ( E_m(\xi) + \Delta_m(\xi) ) \| \varphi \|^2  - \Delta_m(\xi) \| \Gamma(j_1) \varphi \|^2 - o(L^0) \| \varphi \|_{H_m(\xi)}^2 ,
\end{align}
where we used that $\inn{ J \varphi , ( \one \otimes P_\Omega) J \varphi } = \| \Gamma(j_1) \varphi \|^2$.  Let $\lambda \in \sigma_{\rm ess}(H_m(\xi))$. Then there
exists a normalized sequence $\psi_n$, $n \in \N$, converging weakly to zero such that \[\lim_{n \to \infty} \| (H_m(\xi) - \lambda ) \psi_n \|  =  0.\] Thus,
$$
\inn{ \psi_n, H_m(\xi) \psi_n } \geq E_m(\xi) + \Delta_m(\xi)   - \Delta_m(\xi) \| \Gamma(j_1) \psi_n \|^2 - o(L^0) \| \psi_n \|_{H_m(\xi)}^2  .
$$
Taking the limit $n \to \infty$ we find
$$
\| \Gamma(j_1) \psi_n \|^2 =  \inn{(1 + H_{f,m}) \psi_n ,(1 + H_{f,m})^{-1}  \Gamma(j_1^2)  \psi_n }  \to 0 ,
$$
since $ (1 + H_{f,m})^{-1} \Gamma(j_1^2)$ is compact (it is compact on every finite particle space and, since  $m > 0$, it is given by  $\lim_{n \to \infty} (1 + H_{f,m})^{-1} \Gamma(j_1^2) 1_{N \leq n}$
in operator norm). Thus we find
$$
\lambda \geq E_m(\xi) + \Delta_m(\xi) + o(L^0) (\lambda^2 + 1 ).
$$
Taking $L \to \infty$ yields the claim.
\end{proof}

\begin{proposition} \label{prop:deltpos}
Let  $\rho \in L^2(\R^3;(|k|+|k|^{-1}) dk)$ with $ \rho = \rho(- \cdot)$. Let $ e \in \R$ and $m > 0$, and  suppose  \eqref{eq:eineq} holds.
 Then  $\Delta_m(\xi) > 0$, whenever  $|\xi| \leq 1$.
\end{proposition}
\begin{proof} First we show that  the function $\xi \mapsto E_m(\xi)$ has  the properties
\begin{itemize}
\item[(i)] $E_m(0) \leq E_m(\xi)$,
\item[(ii)] $E_m(\xi) \leq \frac{1}{2} \xi^2 + E_m(0) $,
\item[(iii)]  $G_m : \xi \mapsto  \frac{1}{2} \xi^2 - E_m(\xi)$ is convex.
\end{itemize}
Property (i) follows from the assumption, (iii) follows since the pointwise supremum of a set of convex functions is convex.
The symmetry $E_m(-\xi) = E_m(\xi)$, which follows
form Lemma \ref{lem:symmlemm}  below,  implies  $G_m(-\xi) = G_m(\xi)$.  Thus by convexity  $G_m(0) \leq G_m(\xi)$, which implies (ii).
It follows from a lemma about convex functions (see Lemma A2 in [LossMiyaoSpohn]) that properties (i)-(iii) imply
$$
E_m(\xi - k ) - E_m(\xi) \geq \left\{ \begin{array}{ll} - |k| | \xi | + \frac{1}{2} k^2  & \text{,  if } |k| \leq |\xi | , \\ -\frac{1}{2} \xi^2  & \text{, if } |k| \geq |\xi| .\end{array} \right.
$$
Thus we find using  $\omega_m(k) > |k|$
\begin{align*}
E_m(\xi - k ) - E_m(\xi) + \omega_m(k)
  & >  \left\{ \begin{array}{ll} - |k| | \xi | + |k|   & \text{,  if } |k| \leq |\xi | , \\ -\frac{1}{2} \xi^2 + |\xi|  & \text{, if } |k| \geq |\xi| , \end{array} \right.  \\
 &  \geq  0 ,
\end{align*}
provided  $|\xi | \leq  1$.
\end{proof}

\begin{lemma} \label{lem:symmlemm}  Let $D_{\lambda,\lambda'}(k) = \varepsilon_\lambda(k) \cdot (- \varepsilon_{\lambda'}(-k))$. Define $I : L^2(\Z_2 \times \R^3 ) \to L^2(\Z_2 \times \R^3 )$
by  $(I \psi)(\lambda,k) = \sum_{\lambda'} D_{\lambda,\lambda'}(k) \psi(\lambda',-k)$ for $\psi \in L^2(\Z_2 \times \R^3)$. Then $I$ and $\Gamma(I)$ are  unitary operators.
 For   $m > 0$ and   $\rho \in L^2(\R^3;(|k| + |k|^{-1}  ) dk )$,
 we have  
 $\Gamma(I)  H_m(\xi)  \Gamma(I)^*  = H_m(-\xi)$.
\end{lemma}
\begin{proof} It is straightforward to verify that $I$ is unitary.  Hence $\Gamma(I)$ is also unitary.
It is easy to see that   $\Gamma(I) H_f \Gamma(I)^* =H_f  $ and $\Gamma(I) P_f \Gamma(I) ^*= - P_f$. Using the properties of the polarization vectors and the parity symmetry of
$\rho$, an  elementary calculation shows that   $\Gamma(I) A_j \Gamma(I)^* = - A_j$.
The claim of the lemma now follows.
\end{proof}

\bibliography{existgsshortsubmitarxiv}

\end{document}